\documentclass[runningheads]{llncs}
\usepackage[utf8]{inputenc}
\usepackage{graphicx,xcolor,enumerate}
\usepackage[font=small]{caption}
\usepackage{subcaption}
\usepackage{booktabs}
\usepackage{units}
\usepackage{wrapfig}
\usepackage{amssymb, amsmath}%
\usepackage{fullpage}
\usepackage{xspace}
\usepackage[english]{babel}
\usepackage{hyperref}
\usepackage{enumerate}
\usepackage{paralist}

\usepackage{booktabs}
\usepackage{multirow}
\usepackage{todonotes}

\newcommand{\probname}[1]{\textsc{#1}\xspace}

\newcommand{\MaxTotal}{\probname{MaxTotal}}

\newcommand{\GreedyMax}{\textsf{GreedyMax}\xspace}
\newcommand{\GreedyLowCost}{\textsf{GreedyLowCost}\xspace}
\newcommand{\GreedyBestRatio}{\textsf{GreedyBestRatio}\xspace}

\newcommand{\GM}{\textsf{GM}\xspace}
\newcommand{\GLC}{\textsf{GLC}\xspace}
\newcommand{\GBR}{\textsf{GBR}\xspace}
\newcommand{\ILP}{\textsf{ILP}\xspace}
\newcommand{\QAPX}{\textsf{QAPX}\xspace}

\newcommand{\QGM}{\textsf{QGM}\xspace}
\newcommand{\QGLC}{\textsf{QGLC}\xspace}
\newcommand{\QGBR}{\textsf{QGBR}\xspace}

\newcommand{\totact}{total act.\xspace}

\newcommand{\ie}{i.\,e.}
\newcommand{\eg}{e.\,g.}

 \newcommand{\quarter}{1/4} %
 \newcommand{\eight}{1/8}  %

\let\doendproof\endproof
\renewcommand{\endproof}{\hfill\qed\doendproof}

\graphicspath{{fig/},{fig/examples/}}

\title{Evaluation of Labeling Strategies for Rotating Maps} \author{Andreas
  Gemsa, Martin N{\"o}llenburg, Ignaz Rutter }
\institute{Institute of Theoretical Informatics, Karlsruhe Institute
  of Technology, Germany }

\begin{document}
 \maketitle

\begin{abstract}
	We consider the following problem of labeling points in a dynamic map that allows rotation. We are given a set of points in the plane labeled by a set of mutually disjoint labels, where each label is an axis-aligned rectangle attached with one corner to its respective point. We require that each label remains horizontally aligned during the map rotation and our goal is to find a set of mutually non-overlapping \emph{active} labels for every rotation angle $\alpha \in [0,2\pi)$ so that the number of active labels over a full map rotation of $2\pi$ is maximized.
	
	We discuss and experimentally evaluate several labeling models that define additional consistency constraints on label activities in order to reduce flickering effects during monotone map rotation. We introduce three heuristic algorithms and compare them experimentally to an existing approximation algorithm and exact solutions obtained from an integer linear program. Our results show that on the one hand low flickering can be achieved at the expense of only a small reduction in the objective value, and that on the other hand the proposed heuristics achieve a high labeling quality significantly faster than the other methods. 
\end{abstract}

\section{Introduction}
Dynamic digital maps, in which users can navigate by continuously
zooming, panning, or rotating their personal map view, opened up a new
era in cartography and geographic information science (GIS) from
professional applications to personal mapping services on mobile
devices. The continuously animated map view adds a temporal dimension
to the map layout and thus many traditional algorithms for static maps
do not extend easily to dynamic maps. Despite the popularity and
widespread use of dynamic maps, relatively little attention has been
paid to provably good or experimentally evaluated algorithms for
dynamic maps.

In this paper we consider \emph{dynamic map labeling} for points, i.e., the
problem of deciding when and where to show labels for a set of point
features on a map in such a way that visually distracting effects
during map animation are kept to a minimum.
In particular, we study
rotating maps, where the mode of interaction is restricted to changing
the map orientation, e.g., to be aligned with the travel direction in a car navigation system.

\begin{figure}[tb]
  \centering
  \includegraphics[height=.2\textheight]{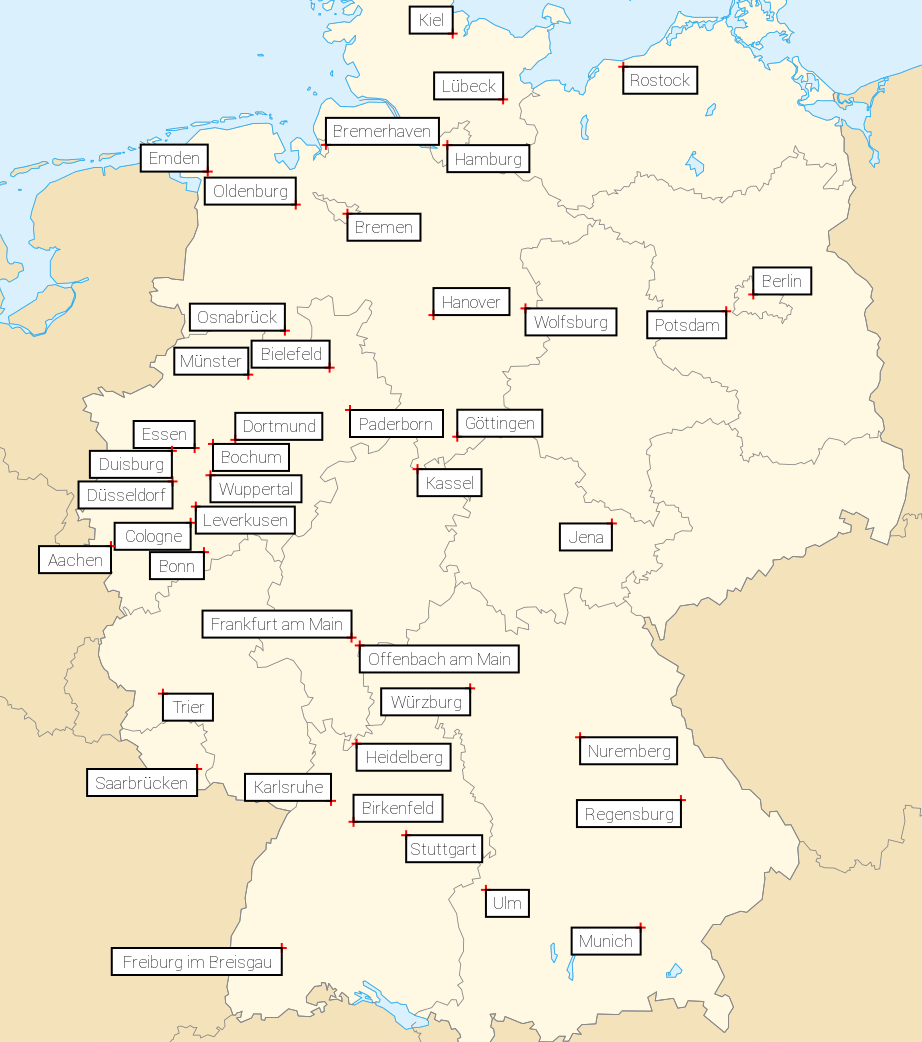}
  \hspace{.3cm}
  \includegraphics[height=.2\textheight]{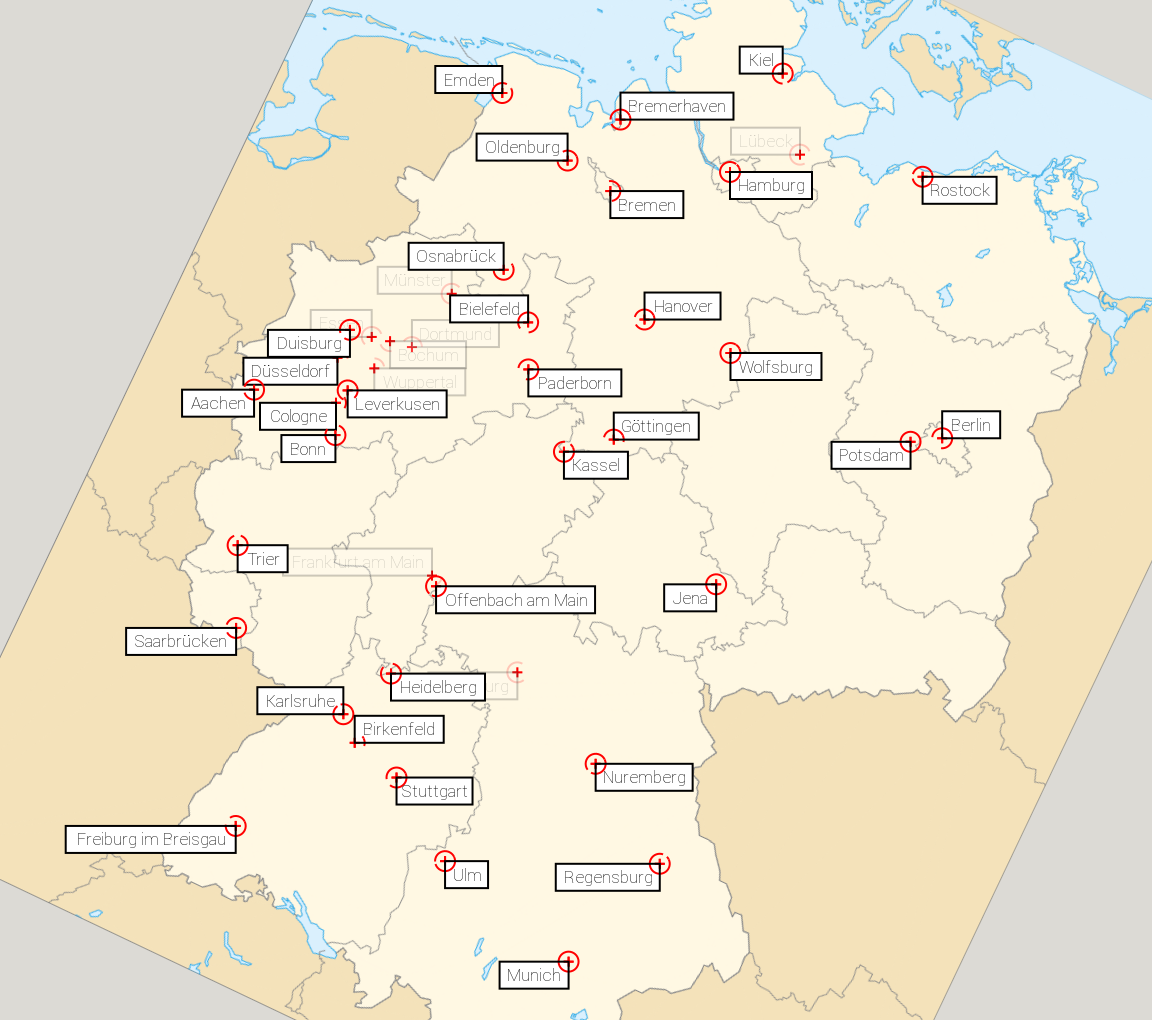}
  \caption{Instance with 43 labeled cities in Germany. Input labeling
    (left), rotated by $\sim 25^\circ$ (right).}%
  \footnotesize{Background picture is in public domain. Retrieved from Wikpedia
    \href{http://commons.wikimedia.org/wiki/File:Germany_localisation_map_2008.svg}{[Link]}}
  \label{fig:example}
\end{figure}

Been et al.~\cite{bdy-dml-06,bnpw-oarcd-10} defined a set of \emph{consistency desiderata} for labeling zoomable dynamic maps, which include that (i) labels do not \emph{pop} or \emph{flicker} during monotone zooming, (ii) labels do not \emph{jump} during the animation, and (iii) the labeling only depends on the current view and not its history. In our previous paper~\cite{gnr-clrm-w-11}, we adapted the consistency model of Been et al.\ to rotating maps, showed NP-hardness and other properties of consistent labelings in this model, and provided efficient approximation algorithms.

Similar to the (NP-hard) label number maximization problem in static map labeling~\cite{fw-ppalm-91}, the goal in dynamic map labeling is to maximize the number of visible or \emph{active} labels integrated over one full rotation of $2\pi$. The value of this integral is denoted as the \emph{total activity} and defines our objective function. Figure~\ref{fig:example} shows an example seen from two different angles. Without any consistency restrictions, we can select the active labels for every rotation angle $\alpha \in [0,2\pi)$ independently of any other rotation angles. Clearly, this may produce an arbitrarily high number of flickering effects that occur whenever a label changes from active to inactive or vice versa. Depending on the actual consistency model, the number of flickering events per label is usually restricted to a very small number. Our goal in this paper is to evaluate several possible labeling strategies, where a labeling strategy combines both a consistency model and a labeling algorithm. First, we want to evaluate the loss in total activity caused by using a specific consistent labeling model rather than an unrestricted one. Second, we are interested in evaluating how close to the optimum total activity our proposed algorithms get for real-world instances in a given consistency model.

\paragraph{Related Work.}
Most previous work on dynamic map labeling covers maps that allow
panning and
zooming, e.g.,~\cite{bdy-dml-06,bnpw-oarcd-10,nps-dosbl-10,okf-dlu-09,vtw-tcrld-2012};
there is also some work on labeling dynamic points in a static
map~\cite{bg-aafm-12,bg-lmpwtblslo-13}. As mentioned above, the
dynamic map labeling problem for rotating maps has first been
considered in our previous paper~\cite{gnr-clrm-w-11}. We
introduced a consistency model, and proved NP-completeness even for
unit-square labels.  For unit-height labels we described an efficient
\quarter-approximation algorithm as well as a PTAS. Yokosuka and
Imai~\cite{yi-ptalsmrm-13} considered the label size maximization
problem for rotating maps, where the goal is to find the maximum font
size for which all labels can be constantly active during
rotation. Finally Gemsa et al.~\cite{gnn-tbdml-isaac-13} studied a
trajectory-based labeling model, in which a locally consistent
labeling for a viewport moving along a given smooth trajectory needs
be computed. Their model combines panning and rotation of the map
view.

\paragraph{Our Contribution.}
In this paper we take a practical point of view on the dynamic map
labeling problem for rotating maps. In Section~\ref{sec:prelim} we
formally introduce the algorithmic problem and discuss our original
rather strict consistency model~\cite{gnr-clrm-w-11}, as well as two
possible relaxations that are interesting in
practice. Section~\ref{sec:algos} summarizes the known
\quarter-approximation algorithm~\cite{gnr-clrm-w-11}, introduces
three greedy heuristics (one of which is a \eight-approximation for
unit square labels), and presents a formulation as an integer linear
program (ILP), which provides us with optimal solutions against which
to compare the algorithms. Our main contribution is the experimental
evaluation in Section~\ref{sec:exper-eval}. We extracted several
real-world labeling instances from OpenStreetMap data and make them
available as a benchmark set. Based on these data, we evaluate both
the trade-off between the consistency and the total activity, and the
performance of the proposed labeling algorithms. The experimental
results indicate that a high degree of labeling consistency can be
obtained at a very small loss in activity. Moreover, our greedy
algorithms achieve a high labeling quality and outperform the running
times of the other methods by several orders of magnitude. We conclude
with a suggestion of the most promising labeling strategies for
typical use cases.

\section{Preliminaries}\label{sec:prelim}

In this section we describe a general labeling model for rotating maps with axis-aligned rectangular
labels. This model extends our earlier model~\cite{gnr-clrm-w-11}.

Let $M$ be an (abstract) map, consisting of a set $P=\{p_1, \dots, p_n\}$ of points in the plane
together with a set $L = \{\ell_1, \dots, \ell_n\}$ of pairwise disjoint, closed, and axis-aligned
rectangular labels in the plane. Each point $p_i$ must coincide with a corner of its corresponding label $\ell_i$;
we denote that corner (and the point $p_i$) as the \emph{anchor} of label $\ell_i$. Since each label
has four possible positions with respect to $p_i$ this widely used model is known in the literature as
the 4-position model (4P)~\cite{fw-ppalm-91}.

As $M$ rotates, each label $\ell_i$ in $L$ must remain horizontally aligned and anchored at~$p_i$.
Thus, new label intersections form and existing ones disappear during the rotation of $M$. We take the
following alternative perspective on the rotation of $M$. Rather than rotating the points, say
clockwise, and keeping the labels horizontally aligned we may instead rotate each label
counterclockwise around its anchor point and keep the set of points fixed. Both
rotations are equivalent in the sense that they yield exactly the same intersections of labels and
occlusions of points.

We consider all rotation angles modulo~$2\pi$.  For convenience we
introduce the interval notation $[a,b]$ for any two angles $a,b \in
[0, 2\pi]$. If $a \le b$, this corresponds to the standard meaning of
an interval, otherwise, if $a > b$, we define $[a,b] := [a,2\pi] \cup
[0,b]$.  For simplicity, we refer to any set of the form~$[a,b]$ as an
interval. We  define the length of an interval $I = [a, b]$ as $|I| =
b-a$ if $a \le b$ and $|I| = 2\pi - a + b$ if 
$a > b$.

A \emph{rotation} of $L$ is defined by a rotation angle $\alpha \in
[0,2\pi)$.  We define $L(\alpha)$ as the set of all labels, each
rotated by an angle of~$\alpha$ around its anchor point. A
\emph{rotation labeling} of $M$ is a function $\phi\colon L \times
[0,2\pi) \rightarrow \{0,1\}$ such that $\phi(\ell,\alpha) = 1$ if
label $\ell$ is visible or \emph{active} in the rotation of $L$ by
$\alpha$, and $\phi(\ell,\alpha) = 0$ otherwise.  We call a
labeling~$\phi$ \emph{valid} if, for any rotation $\alpha$, the set of
labels $L_\phi(\alpha) = \{\ell \in L(\alpha) \mid
\phi(\ell,\alpha)=1\}$ consists of pairwise disjoint labels.  If two
labels $\ell$ and $\ell'$ in $L(\alpha)$ intersect, we say that they
have a (soft) \emph{conflict} at $\alpha$, i.e., in a valid labeling
at most one of them can be active at $\alpha$. We define the set
$C(\ell,\ell') = \{\alpha \in [0,2\pi) \mid \ell \text{ and } \ell'
\text{ are in conflict at } \alpha\}$ as the \emph{conflict set} of
$\ell$ and $\ell'$.  Further, we call a contiguous range in $C(\ell,
\ell')$ a \emph{conflict range}. The begin and end of a maximal
conflict range are called \emph{conflict events}.

For a label
$\ell$ we call each maximal interval $I \subseteq [0, 2\pi)$
with~$\phi(\ell,\alpha) = 1$ for all~$\alpha \in I$ an \emph{active
  range} of label~$\ell$ and define the set $A_\phi(\ell)$ as the set of all active ranges of $\ell$ in $\phi$.   We call an active range where
both boundaries are conflict events a \emph{regular} active range.
Our optimization goal is to find a valid labeling $\phi$ that shows a maximum number of labels integrated over one full rotation from $0$ to $2\pi$. The value of this integral is called the \emph{total activity} $t(\phi)$ and can be computed as $t(\phi) = \sum_{\ell \in L} \sum_{I \in A_\phi(\ell)} |I|$. The problem of optimizing $t(\phi)$ is called  \emph{total activity maximization problem} (\MaxTotal).

A valid labeling is not yet consistent in terms of the definition of
Been et al.~\cite{bdy-dml-06,bnpw-oarcd-10}: while labels clearly do
not jump and the labeling is independent of the rotation history,
labels may still \emph{flicker} multiple times during a full rotation
from $0$ to $2\pi$, depending on how many active ranges they have in
$\phi$.  In the most restrictive consistency model, which avoids
flickering entirely, each label is either active for the full rotation
$[0,2\pi)$ or never at all. We denote this model as
\emph{0/1-model}. In our previous paper~\cite{gnr-clrm-w-11} we
defined a rotation labeling as consistent if each label has only a
single active range, which we denote here as the \emph{1R-model}. This
immediately generalizes to the \emph{$k$R-model} that allows at most
$k$ active ranges for each label. Analogously, the unrestricted model,
\ie, the model without restrictions on the number of active ranges per
label, is denoted as the \emph{$\infty$R-model}.

We may apply another restriction to our consistency models, which is
based on the occlusion of anchors.  Among the conflicts in set
$C(\ell,\ell')$ we further distinguish \emph{hard conflicts}, i.e.,
conflicts where label $\ell$ intersects the anchor point of label
$\ell'$.  If a labeling $\phi$ sets $\ell$ active during a hard
conflict with $\ell'$, the anchor of $\ell'$ is occluded. This may be
undesirable in some situation in practice, e.g., if every point in $P$
carries useful information in the map, even if it is unlabeled. Thus
we may optionally require that $\phi(\ell,\alpha) = 0$ during any hard
conflict of a label $\ell$ with another label $\ell'$ at angle
$\alpha$. Note that we can include other obstacles (e.\,g., important
landmarks on a map) which must not be occluded by a label in form of hard
conflicts. Note that a soft conflict is always a label-label conflict,
while a hard conflict is always a label-point conflict (in our
definition every label-point conflict induces also a label-label conflict). We showed earlier~\cite{gnr-clrm-w-11} that \MaxTotal is NP-hard in the 1R-model avoiding hard conflicts and presented approximation algorithms.

\section{Algorithmic Approaches}\label{sec:algos}
In this section we describe five algorithmic approaches for computing
consistent active ranges that we evaluate in our experiments.
Section~\ref{sec:heuristics-maxtotal} describes three simple greedy
heuristics. Then, we sketch in Section~\ref{sec:.25-apprx} our
\quarter-approximation algorithm which we described in more detail
in~\cite{gnr-clrm-w-11}.  It is based on the shifting technique by
Hochbaum and Maass~\cite{hm-ascpp-85}, where instances are decomposed
into small independent cells that are then solved optimally. Finally,
we give an ILP formulation in Section~\ref{sec:ilp:maxtotal} that that
we use primarily for evaluating the quality of the other solutions.

\subsection{Greedy Heuristics}
\label{sec:heuristics-maxtotal}
In this section we describe three new greedy algorithms to construct valid and consistent labelings with high total activity. These algorithms are conceptually simple and easy to
implement, but in general we cannot give quality guarantees
for the solutions computed by these algorithms. 

All three greedy algorithms follow the same principle of iteratively
assigning active ranges to all labels. The algorithm first initializes
a set~$L'$ with all labels in~$L$. Then it computes for each label $\ell$
its \emph{maximum active range} $I_{\max}(\ell)$, which is the active range of maximum length $|I_{\max}(\ell)|$ such that (i) $\ell$ is not active while in conflict with another active
label that was already considered by the algorithm, and (optionally) such that (ii) $\ell$ is not active while it has a hard conflict with another label.
Initially the maximum active range of each label is either the full interval $[0,2\pi]$ or the largest range that avoids hard conflicts.
Then the
algorithm repeats the following steps.  It selects and removes a label
$\ell$ from $L'$, assigns it the active range $I_{\max}(\ell)$, and updates
those labels in $L'$ whose maximum active range is affected by
the assignment of $\ell$'s active range. If we consider the $k$R-model with $k > 1$, we keep a counter for the number of selected active ranges and add another copy of $\ell$ with the next largest active range to $L'$ if the counter value is less than $k$.
The three algorithms differ only in the criterion that determines which label is
selected from $L'$ in each iteration.

The first algorithm we propose is called \GreedyMax.  In each step the
algorithm selects the label with the largest maximum active range
among all labels in~$L'$.  Ties are broken arbitrarily. The second
algorithm, \GreedyLowCost, determines for the maximum active range of
each label the cost of adding it to the solution. This means that for
each label $\ell \in L'$ with maximum active range $I_{\max}(\ell)$
the algorithm determines for all labels $\ell' \in L'$ that are in
conflict with $\ell$ during $I_{\max}(\ell)$ by how much their maximum
active range would shrink. The sum of this is the \emph{cost}
$c(\ell)$ of assigning the active range $I_{\max}(\ell)$ to
$\ell$. Among all labels in $L'$ \GreedyLowCost chooses the one with
lowest cost. Finally, the last algorithm, \GreedyBestRatio is a
combination of the two preceding ones. In each step the algorithm
chooses the label $\ell$ whose ratio $|I_{\max}(\ell)|/c(\ell)$ is
maximum among all labels in $L'$.  We conclude with a brief
performance analysis of our algorithms.

\newcommand{\thmgreedy}{In the $k$R-model with constant $k$ the
  algorithm \GreedyMax can be implemented to run in time $O(c n \cdot
  (c + \log n))$ and the algorithms \GreedyLowCost and
  \GreedyBestRatio can be implemented to run in time $O(c n \cdot (c^2
  + \log n))$, where $n$ is the number of labels and $c$ is the
  maximum number of conflicts per label in the input instance.  The
  space consumption of all algorithms is~$O(c n)$.}
\begin{theorem}
\thmgreedy
\label{thm:greedy}
\end{theorem}

\begin{proof}
  We begin by describing the time complexity for \GreedyMax and sketch
  how to adapt the proof for \GreedyBestRatio and \GreedyLowCost.

  For the initialization we need to compute the maximum active range
  for all labels which can be done in $O(cn)$ time. Note that this is
  only necessary in the hard-conflict model, since in the
  soft-conflict model the maximum active range is for all labels is
  the full rotation. To efficiently query for the label with longest
  maximum active range we maintain a maximum heap~$H$ in which we
  store all labels from the set $L'$ with the length of their maximum
  active range as key. Initially we need to add all labels from $L$
  into $H$, which requires $O(n \log n)$ time.

  In each step of the algorithm it first selects the label~$\ell$ with
  maximum active range still left in~$L'$.  Then, it needs to update
  the maximum active range of those labels in $L'$ that have a
  conflict with $\ell$.  Since we maintain all labels contained in
  $L'$ in the heap $H$, we can find and remove the label with maximum
  active range in $L'$ in $O(\log n)$ time.  To determine the new
  maximum active ranges for those labels in $L'$ that are in conflict
  with $\ell$, we conduct a simple linear sweep over $[0, 2\pi)$.
  Note that there are at most $O(c)$ conflict events and hence a
  single sweep requires $O(c)$ time. Finally, we need to update the
  value of the maximum active ranges in the global heap.  This
  requires $O(c\log n)$ time. Thus, a single step in the algorithm
  requires $O(c^2 + c\log n)$ time. Since there are $n$ steps, the
  claimed running time follows.

  For \GreedyLowCost and \GreedyBestRatio we use the same approach,
  but we store the labels with their cost and gain-cost ratio as key,
  respectively. The main difference compared to \GreedyMax is the
  necessity to compute the cost of selecting a maximum active range,
  which can be done straightforwardly in $O(c^2)$ time per
  label. After selecting a label and assigning it its maximum active
  range, the maximum active ranges of at most $O(c)$
  labels still left in $L'$ may change. For these labels, the
  algorithm needs to recompute the cost of selecting their maximum
  active ranges and update these values in the maximum heap.  In total
  a single step for \GreedyLowCost or \GreedyBestRatio requires $O(c^3
  + c \log n)$ time. The claimed time complexity follows.  The
  required space is dominated by the storage required to store for
  each label its relevant conflict events.  This takes~$O(cn)$ space.
\end{proof}

By using a more efficient encoding of the maximal disjoint intervals
that can be assigned to label~$\ell_i$ in a heap, the running time of
\GreedyMax can be further improved to~$O(cn \log n)$.

\begin{theorem}
  \GreedyMax can be implemented with time complexity $O(cn \log n)$
  and space requirement in $O(cn)$.
\end{theorem}

\begin{proof}
  Similar to the proof of Theorem~\ref{thm:greedy}, we maintain a
  global heap that contains the maximum possible active range for each
  label.  In the following we describe the modified representation of
  the possible active ranges for the individual labels.

  We maintain for \emph{each} label~$\ell_i \in L$ a maximum
  heap~$H_i$ that maintains all maximal disjoint intervals the
  label~$\ell_i$ can be assigned as active range. Further, \emph{each}
  label also maintains a balanced binary tree $T_i$ on the same set of
  intervals, \ie, we store the left and right endpoints of the
  intervals in $T_i$. Should one of the intervals span over $2\pi$, we
  split it into two intervals at $2\pi$.  Since there are $O(c)$
  conflict events per label, both $T_i$ and $H_i$ can contain at most
  $O(c)$ elements.

  For the initialization we first need to compute all possible maximal
  active ranges for each label which can be done in $O(cn)$
  time. Initializing the global heap can be done as before in $O(n
  \log n)$ time, and since each label specific heap~$H_i$ and binary
  tree $T_i$ can contain at most $O(c)$ elements, their initialization
  requires $O(cn \log c)$ time in total.

  Now, when the algorithm has chosen a label $\ell$ with maximum
  active range~$A(\ell)$ with left and right endpoints~$a$ and~$b$,
  respectively, it needs to update the maximum active range of each
  label $\ell_i$ that is in conflict with $\ell$.  So, for each label
  $\ell_i$ in conflict with~$\ell$, we query $\ell_i$'s binary search
  tree $T_i$ with $a$ and $b$. This allows us to obtain all $k$
  maximal intervals of $\ell_i$ that partly overlap with, or are
  completely contained in $A(\ell)$. The intervals completely
  contained in $A(\ell)$ need to be removed from $T_i$ and $H_i$,
  while the intervals (at most two) that are only partly contained in
  $A(\ell)$ are shrunk or split.

  The query on the binary search tree requires $O(k + \log c)$ time,
  where $k$ is the number of reported intervals.  Updating both $T_i$
  and $H_i$ requires then $O(\log c)$ time for the intervals that get
  shrunk (or split).  Deleting the elements is more costly, but since
  in both the binary tree as well as the heap there can be at most
  $O(c)$ elements, and we insert and remove each element at most once,
  we can conclude that inserting and deleting requires per binary tree
  and heap each in total at most $O(c \log c)$ time. Since there are
  $O(n)$ heaps and trees, this requires in total $O(cn \log c)$ time.

  We now summarize these results. The initialization can be done in
  $O(cn \log n)$ time. In each step we require $O(\log n)$ time to
  determine and remove the label $\ell$ with maximum active range from
  the global heap.  Then, we shrink or split the intervals of labels
  that are in conflict with $\ell$ in $O(c \log c)$ time.  We need to
  update the global heap since the maximum active range of the $O(c)$
  labels that are in conflict with $\ell$ might have changed.  This
  update requires $O(c \log n)$ time.  Finally, we need to repeat this
  step $O(n)$ times, yielding $O(cn (\log c + \log n)) = O(cn \log n)$
  time.  As stated above, the insertion and deletion into the label's
  own binary trees and maximum heaps requires in total $O(cn \log n)$
  time, yielding a total time complexity of $O(cn \log n)$.  The space
  consumption is dominated by the~$O(n)$ trees, each of size~$O(c)$.
  Thus the algorithm requires~$O(cn)$ space.
\end{proof}

\paragraph{Approximation for Unit Square Labels.}
Although it was stated at the beginning of the
Section~\ref{sec:heuristics-maxtotal} that, in general, we cannot
prove any quality criteria for the presented algorithms, for the
special case that labels are unit squares, we can show that \GreedyMax
is an approximation algorithm with approximation ratio~$1/8$.  This is
due to a result by Gemsa et al.~\cite[Lemma 1 and Lemma 2]{gnn-tbdml-isaac-13}
in which the authors consider a similar problem. The results of both
Lemmas imply the following.

\begin{lemma}[\hspace{0.001mm}\cite{gnn-tbdml-isaac-13}, Lemmas 1 and 2]
  \label{lem:gemsaetal}
  Let~$L$ be a set of unit square labels and let~$\ell \in L$.  For
  any rotation angle~$\alpha \in [0, 2\pi)$ there are at most eight
  pairwise disjoint unit square labels in~$L(\alpha)$ that overlap
  with~$\ell$.
\end{lemma}

Since, in each step of the \GreedyMax algorithm, we select the label
$\ell$ with maximum active range, we can conclude by
Lemma~\ref{lem:gemsaetal} that the cost of choosing $\ell$ is at most
eight times its maximum active range.  We summarize this in the
following.

\begin{corollary}
  \GreedyMax is a \eight-approximation algorithm for \MaxTotal in the
  case that all labels are unit squares.
\end{corollary}

\subsection{\quarter-Approximation Algorithm}
\label{sec:.25-apprx}

Our previous \quarter-approximation~\cite{gnr-clrm-w-11} is based on
the \emph{line stabbing} or \emph{shifting} technique by Hochbaum and
Maass~\cite{hm-ascpp-85}, which has been applied before to
\emph{static} labeling problems for (non-rotating) unit-height
labels, \eg,~\cite{aks-lpmis-98,bg-aafm-12,ksw-pslsl-99}. We sketch the idea
of the algorithm for unit squares, but note that it applies in a
similar way to rectangular labels with bounded aspect ratio and that
it can be generalized to a PTAS~\cite{gnr-clrm-w-11}.  The first step
of the algorithm is to subdivide the plane into square grid cells with
side length~$2\sqrt{2}$. Each cell is addressed by a row and column
index. We obtain a partition of the initial instance into four
different subsets by deleting the points in every other row and every
other column. In each subset the cells now have a distance to each
other that is at least $2\sqrt{2}$ and thus, labels that are in
different cells \emph{cannot} intersect each other; see
Figure~\ref{fig:qapx}.  Note that for hard conflicts we need to
consider the area within distance $\sqrt{2}$ of each cell.  If we
solve all four subsets optimally, at least one of the solutions is a
\quarter-approximation for the entire instance due to the pigeon-hole
principle.

\begin{figure}[tb]
  \centering

  \includegraphics[page=3, width=.33\textwidth]{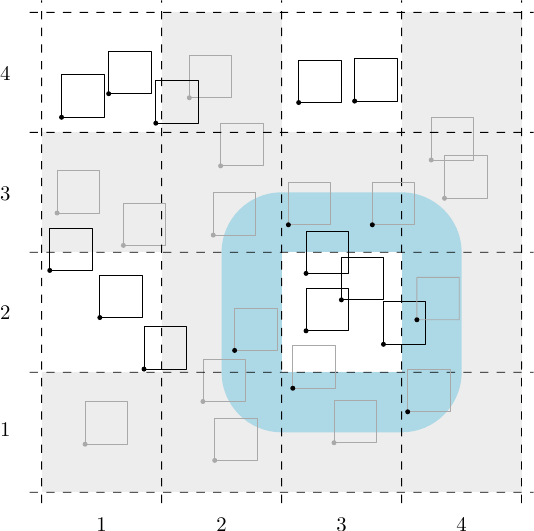}
  \caption{Plane divided into grid cells. Cells of one subinstance are marked white. The blue area has distance at most $\sqrt{2}$ to the cell~(3,~2).}
  \label{fig:qapx}
\end{figure}

To find the optimal solution for a single grid cell we observed that
there can only be a constant number of labels in each grid
cell~\cite{gnr-clrm-w-11}. Further, we observed that this implies that
also the number of conflict events in each grid cell is constant which
int combination with \cite[Lemma 4]{gnr-clrm-w-11} means that the
optimal solution can be obtained by a simple brute force approach in
$O(1)$ time per cell, which results in an overall running time of $O(n
\log n)$.

\subsection{Integer Linear Program}
\label{sec:ilp:maxtotal}
In this section we present an ILP-based approach to find optimal
solutions for \MaxTotal. This is justified since \MaxTotal is NP-hard
and we cannot hope for an efficient algorithm unless $P = NP$. We note
that the same ILP formulation can also be used in the
\quarter-approximation algorithm to compute an optimal solution within
each grid cell. The ILP formulation given is similar to that of Gemsa
et al.~\cite{gnn-tbdml-isaac-13}.

The key idea of the ILP presented here is to determine regular active
ranges induced by the ordered set of all conflict events. Our model
contains for each label $\ell$ and each interval $I$ a binary decision
variable, which indicates whether or not $\ell$ is active during~$I$.
We add constraints to ensure that (i) no two conflicting labels are
active at the same time within their conflict range and (ii) at most
$k$ disjoint contiguous active ranges can be selected for each label
as required in the $k$R-model.

\paragraph{Model.} For simplicity we assume in this section that the
length of each conflict range is strictly larger than~0. This
assumption is not essential for our ILP formulation, but makes the
description easier.

Let $E$ be the ordered set of conflict events that also contains~0
and~$2\pi$, and let $E[j]$ be the interval between the $j$-th and the
$(j+1)$-th element in $E$. We call such an interval $E[j]$ an
\emph{atomic interval} and always consider its index $j$ modulo
$|E|-1$. For each label $\ell_i \in L$ and for each atomic interval
$E[j]$ we introduce two binary variables $x_i^j$ and $b_i^j$ to our
model. We refer to the variables of the form $x_i^j$ as \emph{activity
  variables}. The intended meaning of $x_i^j$ is that its value is~1
if and only if the label $\ell_i$ is active during the $j$-th atomic
interval; otherwise $x_i^j$ has value~0. We use the binary variables
$b_i^j$ to indicate the start of a new active range and to restrict
their total number to $k$. This is achieved by adding the following
constraints to our model.
\begin{align}
  x_i^j - b_i^j & \leq x_{i}^{j-1} && \forall \ell_i \in L \quad \forall
  j \in \{0, \dots, |E|-2\} 
  \label{eq:ilp:1}\\
  \sum_{0 \leq j \leq |E|-2} b_i^j & \leq k && \forall \ell_i \in L
  \label{eq:ilp:3}
\end{align}
The effect of constraint~\eqref{eq:ilp:1} is that it is only possible
to start a new active range for label $\ell_i$ with atomic interval
$E[j]$ (i.e., $x_i^{j-1} = 0$ and $x_i^j=1$) if we account for that
range by setting $b_i^j = 1$. Due to constraint~\eqref{eq:ilp:3} this
can happen at most~$k$ times per label. We can also allow arbitrarily
many active ranges per label as in the $\infty$R-model by completely
omitting the variables $b_i^j$ and the above constraints.

It remains to guarantee that no two labels can be active when they are
in conflict. This can be done straightforwardly since we can compute
for which atomic intervals two labels are in conflict and we ensure
that not both activity variables can be set to~1. More specifically,
for every pair of labels~$\ell_i, \ell_k$ and for every atomic
interval $j$ during which they are in conflict, we add the constraint
 \begin{equation} x_i^j + x_k^j \leq
  1. \label{eq:ilp:2}\end{equation}

Optionally, incorporating hard conflicts can also be done easily as a
hard conflict simply excludes certain atomic intervals from being part
of an active range.  We determine for each label all such atomic
intervals in a preprocessing step and set the corresponding activity
variables to~0.

Among all feasible solutions that satisfy the above constraints, we
maximize the following objective function: $\sum_{\ell_i \in L} \sum_{0 \leq j
  \leq |E| - 2} x_i^j \cdot |E[j]|$, which is equivalent to the total
activity $t(\phi)$ of the induced labeling $\phi$.

This ILP considers only regular active ranges, since label activities
change states only at conflict events.  However, by~\cite[Lemma
4]{gnr-clrm-w-11}, there always exists an optimal solution that is
regular, and hence we are guaranteed to find a globally optimal
solution.

\paragraph{Minimizing the Number of Active Ranges in the ILP.}

In Section~\ref{sec:prelim} we explained that, in order to reduce
flickering, we require that each label has at most one active range.
However, we might be able to reduce flickering even more by finding
among all optimal solutions the one that has the fewest active ranges.
We can modify our ILP to accommodate for this by modifying the
objective function slightly.  Let~$s > 0$ denote the length of a
shortest atomic interval (recall that all atomic intervals are assumed
to have positive length).  Thus, whenever a label is active during an
atomic interval, the total activity increases by at least~$s$.  To
minimize the number of active ranges, we substract $s/2$ from the
objective function for each active range.  In this way, a solution
with larger total activity is always preferred over a solution with
less total activity, while among two solutions with the same total
activity the one with fewer active ranges has the greater objective
value.  Hence, to minimize the number of active ranges while
maintaining optimality, we modify the objective function to
$\sum_{\ell_i \in L}  \sum_{0 \leq j \leq |E| - 2} \left( x_i^j \cdot
  |E[j]| - b_i^j \cdot s/2 \right)$.  To ensure that if label~$\ell_i$
is active at any time, then at least one of the variables~$b_i^j$
is~$1$, we also add the following constraints.
\begin{align}
  \label{eq:1}
  (|E| - 1) \cdot \sum_{j=0}^{|E|-2} b_i^j \ge \sum_{j=0}^{|E|-2} x_i^j && \forall \ell_i \in L
\end{align}
Without this constraint, it would be possible to have a label~$\ell_i$
active for the whole range~$[0,2\pi)$ but with~$b_i^j = 0$ for all~$j
\in \{0,\dots,|E|-2\}$.  We note that it follows from the proof of
Lemma 4 in~\cite{gnr-clrm-w-11} that also for the
modified problem there always exists an optimal solution that is
regular.  Hence an optimal solution to the above ILP is indeed a
global optimum.

We have now given a complete description of our ILP model, and now
turn towards the analysis of the number of variables and constraints
necessary for our model.  Let~$e$ be the number of conflict events and
$c$ be the maximum number of conflict events per label in a \MaxTotal
instance, respectively.  In the worst case the number of constraints
that ensure that the solution is conflict-free (\ie,
constraint~\eqref{eq:ilp:2}) is $O(c \cdot e)$ per label, whereas we
require only $O(e)$ constraints of the other types of constraints per
label.  We conclude the results of this section in the following
theorem.

\begin{theorem}
  The ILP (1)--(3) solves \MaxTotal and has at most $O(e\cdot n)$
  variables and $O(c\cdot e\cdot n)$ constraints, where $n$ is the
  number of labels, $e$ the number of conflict
  events, and $c$ the maximum number of conflicts per label.
\end{theorem}

\section{Experimental Evaluation}
\label{sec:exper-eval}
In this section we present the experimental evaluation of different
labeling strategies based on the consistency models and algorithms
introduced in Sections~\ref{sec:prelim} and~\ref{sec:algos}. We
implemented our algorithms in C++ and compiled with GCC~4.7.1 using
optimization level \texttt{-O3}. As ILP solver we used Gurobi
5.6\footnote{Gurobi is a commercial ILP solver
  \href{http://www.gurobi.com}{\texttt{www.gurobi.com}}}. The running
time experiments were performed on a single core of an AMD Opteron
2218 processor running Linux 2.6.34.10. The machine is clocked at
$\unit[2.6]{GHz}$, has $\unit[16]{GiB}$ of RAM and $\unit[2 \times
1]{MiB}$ of L2 cache.  Before we discuss our results we introduce the
benchmark instances.  The reported running times in the following are
always measured as wall-clock time (as opposed to pure CPU time).

\subsection{Benchmark Instances} 

Since our labeling problem is immediately motivated by dynamic mapping
applications, we focus on gathering real-world data for the
evaluation. As data source we used the publicly available data
provided by the OpenStreetMap project\footnote{OpenStreetMap
  \href{http://www.osm.org}{\texttt{www.osm.org}}}. We extracted the
latitudes, longitudes and names of \emph{all} cities with a population
of at least $50\,000$ for six countries (France, Germany, Italy,
Japan, United Kingdom, and the United States of America) and created
maps at three different scales.

To obtain a valid labeling instance several additional steps are
necessary. First, the width and height of each label need to be
chosen. Second, we need to map latitude and longitude to the
two-dimensional plane.  Third, recall that the input is a statically
labeled map, and hence we need to compute such a static input
labeling.
For the first issue we used the same font that is used in Google Maps,
\ie, \texttt{Roboto Thin}\footnote{Roboto Font
  \href{http://www.google.com/fonts/specimen/Roboto}{\texttt{www.google.com/fonts/specimen/Roboto}}}. The
dimensions of each label were obtained by rendering the label's
corresponding city name in \texttt{Roboto Thin} with font size~13,
computing its bounding box, and adding a small additional buffer.  For
obtaining two-dimensional coordinates from the latitude and longitude
of each point, we use the popular Mercator projection.  For the map
scales we again wanted to be close to Google Maps. Hence, we used the
Mercator projection (where we approximate the ellipsoid with a sphere
of radius $r = 6371$km) for three different scales (65 pixel
$\widehat{=}$ 20km, 50km, 100km) for each country.  For simplicity we
refer to the scale of 65 pixel $\widehat{=}$ 20km only by 20km (and
likewise for the remaining scales).  The last remaining step was to
compute a valid input labeling.  For this we used the 4P
fixed-position model~\cite{fw-ppalm-91} and solved a simple ILP model
to obtain a weighted maximum independent set in the label conflict
graph, in which any two conflicting label positions are linked by an
edge and weights are proportional to the population.
Table~\ref{tab:exp:input:countries} shows the characteristics of our
benchmark data.

We also obtained much larger instances than the country instances
described in Table~\ref{tab:exp:input:countries}. We chose Berlin, New
York, London, and Paris and extracted the names, as well as longitude
and latitude, of all restaurants in each of these cities from
OpenStreetMap.  To obtain valid input data we conducted the same steps
as described for the country instances with \texttt{Roboto Thin} but with
font size~8 and the three scales 65pixel $\widehat{=}$ 20m, 50m,
100m. We report the number of labels for these instances in
Table~\ref{tab:exp:input:restaurants}.

Both sets of benchmark data can be downloaded from our
website\footnote{Benchmark data set
  \href{http://i11www.iti.kit.edu/projects/dynamiclabeling/}{\texttt{i11www.iti.kit.edu/projects/dynamiclabeling/}}}.

\begin{table}[tb]
  \caption{Number of labels in each country instance, the number of
    labels in the  largest connected component (lcc) and
    the number of connected components (cc) in the conflict graph.}
\label{tab:exp:input:countries}
\centering
\begin{tabular}{@{}lllllll@{}}    \toprule
  \multicolumn{1}{c}{} & \multicolumn{6}{c}{\emph{countries}} \\
  & FR  & DE & GB & IT  & JP &  US \\ 
  \cmidrule(r){2-7}
  \emph{scales} & \multicolumn{6}{c}{\#labels (\#labels in lcc / \#cc) } \\\midrule
  20km         & 86 (12/51) & 52 (20/26)  & 99 (73/19) & 131 (28/48) &
  99 (12/34) &  403 (26/203) \\ 
  50km         & 80 (39/9)  & 43 (39/4)  & 68 (66/2)  & 111 (87/5) &
  80 (69/7) &  359 (88/89) \\
  100km        & 69 (69/1)  & 33 (33/1)  & 37 (37/1)  &  68 (68/1) &
  49 (44/3) &  288 (213/16) \\\bottomrule
\end{tabular}
\end{table}

\begin{table}[tb]
  \caption{Number of labels in each city instance, the number of
    labels in the  largest connected component (lcc) and
    the number of connected components (cc) in the conflict graph.}
\label{tab:exp:input:restaurants}
\center
\begin{tabular}{@{}lrrrrr@{}}    \toprule
  \multicolumn{1}{c}{} & \multicolumn{4}{c}{\emph{cities}} \\
  \cmidrule(r){2-5}
  & Berlin         & New York & London & Paris \\
  \cmidrule(r){2-5}
  \emph{scales} & \multicolumn{4}{c}{\#labels (\#labels in lcc / \#cc) } \\\midrule
  20m         & 2744 (13/2205) & 629 (9/445) & 1661 (40/1044) & 2367 (39/1560) \\
  50m         & 2739 (77/1329) & 621 (50/279)& 1620 (175/569) & 2325 (244/703)\\
  100m        & 2628 (416/634) & 602 (88/143)& 1457 (371/289) & 2141 (952/230)\\\bottomrule
\end{tabular}
\end{table}

\subsection{Evaluation of the Consistency Models}\label{sec:eval-model}

In this section we evaluate the different consistency models
introduced in Section~\ref{sec:prelim}. The models differ by the
admissible number of active ranges per label and the handling of hard
conflicts. We begin by analyzing the effect of limiting the number of
active ranges and consider the five models 0/1, 1R, 2R, 3R, and
$\infty$R, all taking hard conflicts into account. As discussed in
Section~\ref{sec:prelim}, the 0/1-model is flicker-free but expected
to have a low total activity, especially in dense instances. On the
other hand, the $\infty$R-model achieves the maximum possible total
activity in any valid labeling, but is likely to produce a large
number of flickering effects. Still, it serves as an upper bound on
the total activities of the other models. The two most important
quality criteria in our evaluation are (i) the total activity of the
solution, and (ii) the average length of the active
ranges. Unfortunately, it proved too time consuming for our ILP to
solve the city instances in a reasonable time frame. Thus, we chose to
restrict ourselves in this analysis to the country instances described
in Table~\ref{tab:exp:input:countries}.

\begin{table}[tb]
  \caption{Average total activity of the optimal solutions for the country instances with respect to the
    maximum possible objective value with standard deviation in brackets. Instances grouped by
    scale. Additionally we report the average interval length 
    normalized to one full rotation, and for
    the $\infty$R-model the average number of intervals per label.}
\label{tab:exp:nmbintervals:relative}

\begin{subtable}{\textwidth}
\centering
\begin{tabular}{@{}lrrrrrrrrrrrrrr@{}}    \toprule
  \multicolumn{1}{l}{\emph{model}} & \multicolumn{1}{c}{0/1} & &
  \multicolumn{2}{c}{1R} & & \multicolumn{2}{c}{2R} & &
  \multicolumn{2}{c}{3R} & & \multicolumn{2}{c}{$\infty$R} \\
  \cmidrule(r){2-2}  \cmidrule(r){4-5} \cmidrule(r){7-8}
  \cmidrule(r){10-11} \cmidrule(r){13-14}
  \emph{scale} &  \totact & & \totact  & $\sim$len & &  \totact  & $\sim$len & & \totact  & $\sim$len & & $\sim$len & $\sim$intervals \\ \midrule
20km &  54.04\% (12.62)  &   &  94.56\% (2.85)  &  0.76  &   &  99.36\% (0.64)  &  0.56  &   &  99.92\% (0.10)  &  0.47  &   &  0.08  &   19.13  \\
50km &  22.42\% (11.74)  &   &  87.79\% (4.44)  &  0.58  &   &  97.69\% (1.95)  &  0.35  &   &  99.54\% (0.69)  &  0.26  &   &  0.01  & 79.23  \\
100km &  6.19\% (5.06)  &   &  81.01\% (2.15)  &  0.44  &   &  95.83\% (1.56)  &  0.27  &   &  99.24\% (0.36)  &  0.19  &   &  0.01  &   128.40  \\ \bottomrule
\end{tabular}
\caption{Hard-conflict model.}
\label{tab:exp:hc:nmbintervals:relative}
\end{subtable}

\begin{subtable}{\textwidth}
\centering
\begin{tabular}{@{}lrrrrrrrrrrrrrr@{}}    \toprule
  \multicolumn{1}{l}{\emph{model}} & \multicolumn{1}{c}{0/1} & &
  \multicolumn{2}{c}{1R} & & \multicolumn{2}{c}{2R} & &
  \multicolumn{2}{c}{3R} & & \multicolumn{2}{c}{$\infty$R} \\
  \cmidrule(r){2-2}  \cmidrule(r){4-5} \cmidrule(r){7-8}
  \cmidrule(r){10-11} \cmidrule(r){13-14}
  \emph{scale} &  \totact & & \totact  & $\sim$len & &
   \totact  & $\sim$len & & \totact  & $\sim$len & & $\sim$len &
   $\sim$intervals \\ \midrule
20km &  69.18\% (8.14)  &   &  98.17\% (1.20)  &  0.83  &   &  99.79\% (0.19)  &  0.64  &   &  99.98\% (0.02)  &  0.59  &   &  0.54  &     1.80  \\
50km &  49.47\% (5.54)  &   &  94.80\% (1.96)  &  0.69  &   &  99.27\% (0.69)  &  0.44  &   &  99.86\% (0.22)  &  0.37  &   &  0.27  &     5.09  \\
100km &  42.41\% (3.12)  &   &  91.67\% (1.59)  &  0.58  &   &  98.57\% (0.40)  &  0.34  &   &  99.76\% (0.07)  &  0.24  &   &  0.07  &    9.55  \\ \bottomrule
\end{tabular}
\caption{Soft-conflict model.}
\label{tab:exp:sc:nmbintervals:relative}
\end{subtable}
\end{table}

In Table~\ref{tab:exp:hc:nmbintervals:relative} we report the total
activity of the optimal solution for the tested models relative to the
solution in the $\infty$R-model. The results of the instances are
aggregated by scale.  We observe that the total activity of the
0/1-model drops to less than 55\% compared to the optimal solution in
the $\infty$R-model even for the least dense instance at scale 20km
and to only 6\% for a scale of 100km. Hence this model is of very
little interest in practice.

We see a strong increase in the average total activity values already
for the 1R-model compared to the optimal solution in the
0/1-model. For the large-scale instance 20km 1R reaches almost 95\% of
the $\infty$R-model, which has more than 19 times the number of
flickering effects and active ranges of average length shorter by a
factor of $1/9$. For map scales of 50km and 100km, the total
activities drop to 88\% and 81\%, respectively, but at the same time
the number of flickering effects and the average active range lengths
in the $\infty$R model are extremely poor. Thus the 1R-model achieves
generally a very good labeling quality by using only one active range
per label.

Finally, we take a look at the middle ground between the 1R- and the
$\infty$R-models. It turns out that total activity of the 2R-model is
off from the $\infty$R-model by less than 1\% at scale 20km and less
than 5\% at scale 100km, but this increase in activity over the 1R
model comes at the cost of producing twice as many flickering effects
and decreasing the average active range length by 30--40\%. If we
allow three active ranges per label, the total activity increases to
more than 99\% of the upper bound in the $\infty$R-model at all three
scales, while having significantly fewer flickering effects and longer
average active ranges. The activity gain by considering the $k$R-model
for $k>3$ is negligible and the disadvantage of increasing the number
of flickering effects dominates.

When we conduct the same analysis for the 0/1, 1R, 2R, and $\infty$R
model in the soft-conflict model, similar trends can be observed.  We
report in Table~\ref{tab:exp:sc:nmbintervals:relative}, analogous to
Table~\ref{tab:exp:hc:nmbintervals:relative}, the total activity of
the optimal solution for the tested models relative to the solution in
the $\infty$R-model. We observe that the trend in both tables is
similar. However, the 0/1-model performs in general significantly
better than in the hard-conflict model, but the results are still not
suitable from a practical point of view.  Already the 1R model
produces results which are close to the maximum possible solution in
the $\infty$R-model, while the difference in the optimal solutions in
the 2R and 3R-model to the $\infty$R-model are negligible.

We conclude that the 1R-model achieves the best compromise between
total activity value and low flickering, at least for maps at larger
scales with lower feature density. For dense maps the 2R- or even the
3R-models yield near-optimal activity values while still keeping the
flickering relatively low. Going beyond three active ranges per label
only creates more flickering but does not provide noticeable
additional value.

It remains to investigate the impact of hard conflicts.  For this we
apply the 1R-model and compare the variant where all conflicts are
treated equally (soft-conflict model) with the variant where hard
conflicts are disallowed (hard-conflict model).  For this we consider
for each map scale the average relative increase in activity value of
the soft-conflict model over the stricter hard-conflict model. For
20km instances the increase is on average $8.51\%$ (standard deviation
2.91), for the intermediate scale 50km it is on average $19.25\%$
(standard deviation 7.86), and for the small-scale map 100km the
increase reaches on average $31.9\%$ (standard deviation 4.72). These
results indicate that, unsurprisingly, the soft-conflict model
improves the total activity at all scales, and in particular for dense
configurations of point features, where labels usually have several
hard conflicts with nearby features. As discussed before, this
improvement comes at the cost of temporarily occluding unlabeled but
possibly important points. It is an interesting open usability
question to determine user preferences for the two models and the
actual effect of temporary point occlusions on the readability of dynamic maps, but such a user
study is out of scope of this evaluation and left as an interesting
direction for future work.

\subsection{Evaluation of the Algorithms}
\label{sec:exper-eval:algorithms}

In this section we evaluate the quality (total activity) and running
time of the \quarter-approxi\-mation algorithm
and the three greedy heuristics
\GreedyMax, \GreedyLowCost, and \GreedyBestRatio
(Section~\ref{sec:heuristics-maxtotal}), which we abbreviate as \QAPX,
\GM, \GLC, and \GBR, respectively. Additionally, we include the ILP
(Section~\ref{sec:ilp:maxtotal}) as the only exact method in the
evaluation. The ILP is also applied to optimally solve the independent
subinstances in the grid cells created by \QAPX. In our implementation
we heuristically improve the running time of the ILP by partitioning
the conflict graph of the labels into its connected components and
solving each connected component individually; see Table
\ref{tab:exp:input:countries} for the number of labels in the largest
connected component and the number of connected components in the
conflict graph of each instance. For the ILP we set a time limit of
1~hour and restrict the ILP solver to a single thread. The same
restrictions are applied to the ILP when solving the small
subinstances in algorithm \QAPX. By the design of the algorithm, a
solution obtained by \QAPX will consist of many labels that have no
active range, although they could be assigned one (all labels that are
discarded to obtain independent cells have active range set to length
0). To overcome this drawback, we propose a combination of \QAPX with
the greedy algorithms.  More specifically, we apply one of our greedy
algorithms to each of the four solutions computed by the
\quarter-approximation and determine among the four resulting
solutions the best one.
In the following we refer to the combination of the
\quarter-approximation with a greedy algorithm by adding a \textsf{Q}
in front of the greedy algorithm's name (\eg, \QGLC).  We report the
results of the algorithms for both the 1R soft-conflict and the 1R
hard-conflict model, which turned out as a reasonable compromise
between low flickering and high total activity in
Section~\ref{sec:eval-model}.

\begin{figure}
  \begin{subfigure}{0.49\textwidth}
   \centering
   \includegraphics[width =\textwidth, trim = 30 25 30 20, clip]{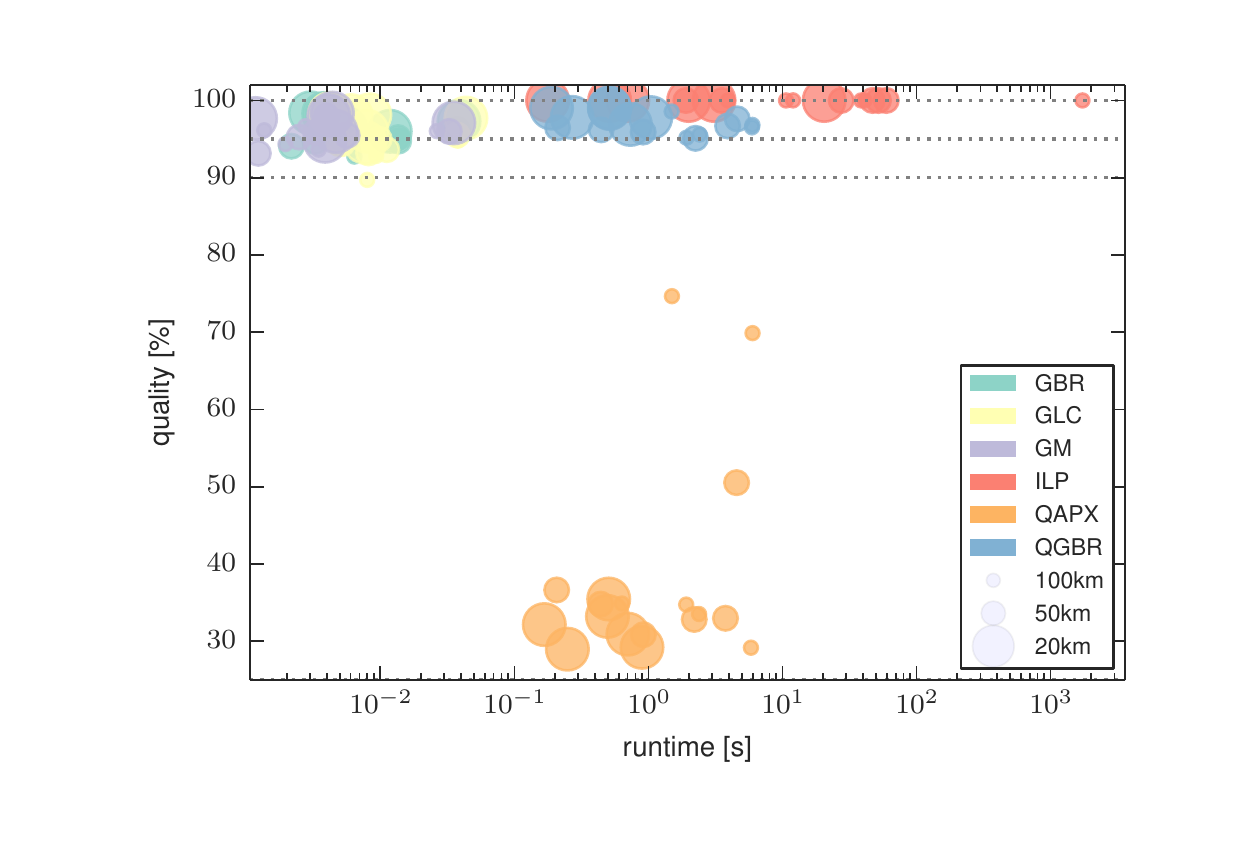} 
   \caption{Hard-conflict model.}
   \label{fig:exp:runtime:countries:hc}
  \end{subfigure}
  \begin{subfigure}{0.49\textwidth}
   \centering
   \includegraphics[width =\textwidth, trim = 30 25 30 20, clip]{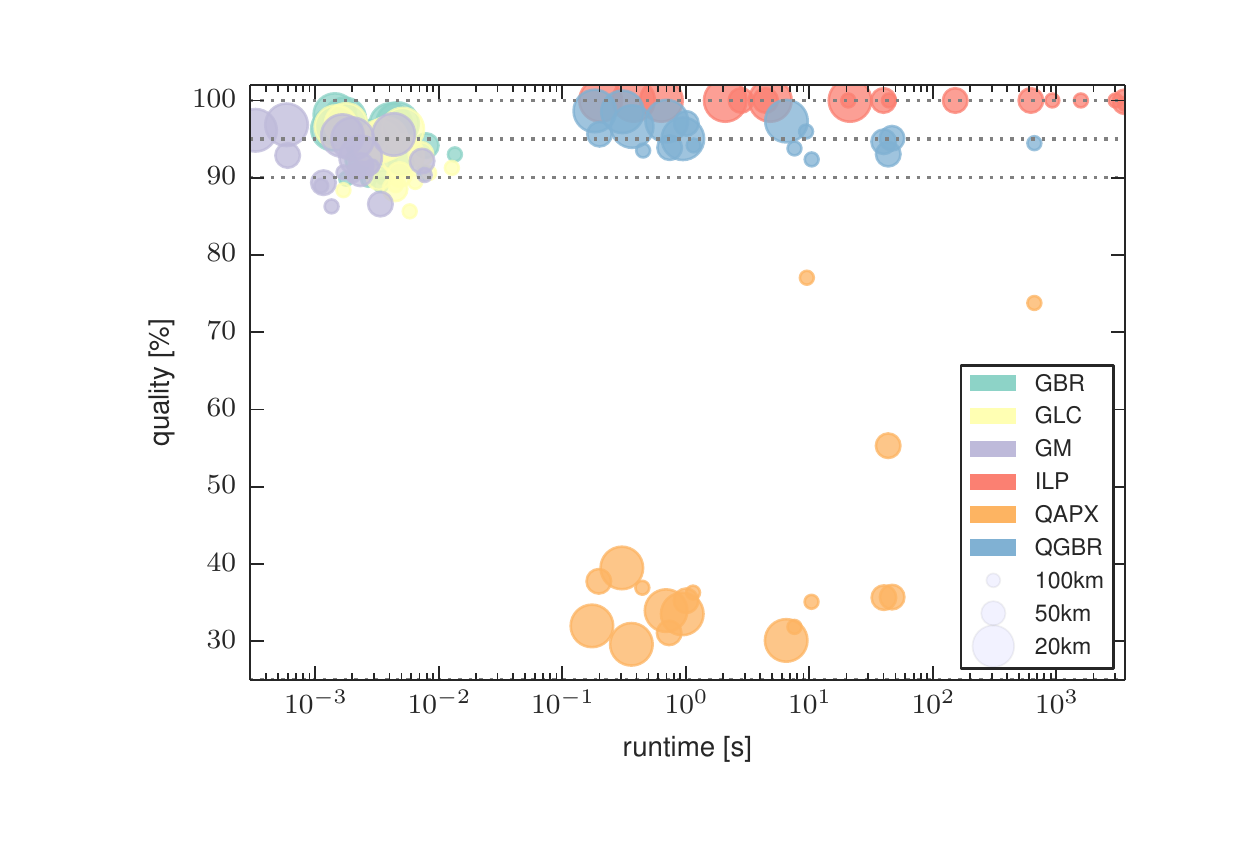} 
   \caption{Soft-conflict model.}
   \label{fig:exp:runtime:countries:sc}
  \end{subfigure}

  \caption{Plot of the runtime (log scale) and quality results of the
    algorithm evaluation for the city instances in the 1R
    hard-conflict and soft-conflict model.}
  \label{fig:exp:runtime:countries}
\end{figure}
We give a general overview of the performance of all
evaluated algorithms as a scatter plot in
Figure~\ref{fig:exp:runtime:countries}. In this scatter plot each
disk represents the result of an algorithm (indicated by color)
applied to a single country instance. The size of the circle indicates
the scale of the instance (the smaller the circle, the smaller the
scale).  We omitted the algorithms \QGM and \QGLC in this plot to
increase readability, because the difference in running time and
quality of the solutions between the three algorithms \QGM, \QGLC, and
\QGBR is negligible and creates extra overplotting.

\paragraph{Soft-conflict model.} In Figure~\ref{fig:exp:runtime:countries:sc} we present an overview of the performance of
the evaluated algorithms in the 1R soft-conflict model.
We observe that the performance of the greedy algorithms is very good
with respect to running time as well as quality of the solutions. As
expected, the total activity of \QAPX is always better than $25\%$,
but generally much worse than for the remaining algorithms. It never
gets close to the solutions produced by the greedy algorithms while
being considerably slower.  However, combining \QAPX with a greedy
algorithm achieves better solutions than greedy algorithms and \QAPX
alone, while the increase in running time over \QAPX is negligible.
Finally, we observe that the \ILP solves the tested instances in a
reasonable time frame.  To obtain the optimal solution, the ILP
required on average 758s, with a median of only 30.65s. However, we
concede that larger instances may require significantly more time to
solve, and there may be a threshold for which the use of the ILP is
infeasible.

\begin{figure}[tb]
  \begin{subfigure}[t]{.55\textwidth}
  \centering
  \includegraphics[height = 3.5cm, trim = 12 15 20 10, clip]{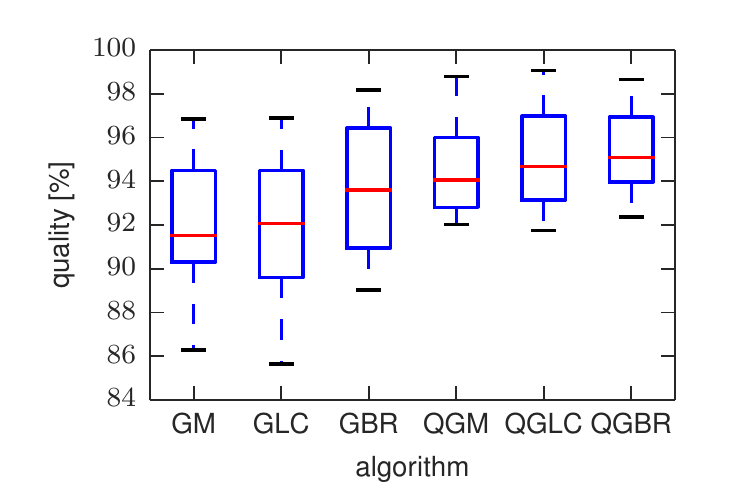} 
  \caption{Quality of the solutions as a percentage of the (optimal)
    ILP solution.}
  \label{fig:exp:runtime:countries:quality}
  \end{subfigure}
  \hfil
  \begin{subfigure}[t]{.39\textwidth}
  \centering
  \includegraphics[height = 3.5cm, trim= 10 15 20 10, clip]{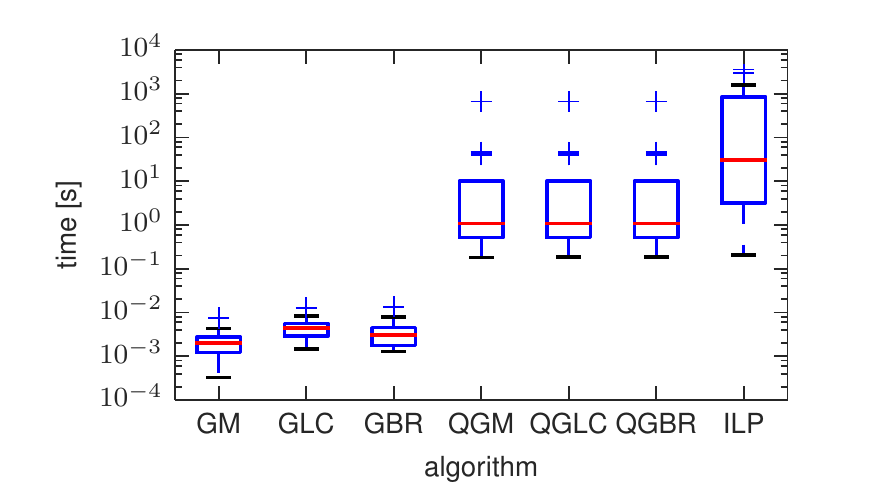} 
  \caption{Runtime (log scale) of the algorithms.}
  \label{fig:exp:runtime:countries:time}
  \end{subfigure}

	\caption{Performance of the greedy algorithms and \QAPX with
          greedy postprocessing in the 1R soft-conflict model.}
  \label{fig:exp:runtime:countries:scmodel}
\end{figure}

We now turn to a more detailed analysis of the two most promising
approaches (i) using the greedy algorithms, and (ii) combining \QAPX
with the greedy algorithms in the soft-conflict model. For a detailed depiction of the
performance of the algorithms with respect to the quality of the
solution see the diagrams in
Figure~\ref{fig:exp:runtime:countries:scmodel}. We observe that among
the three greedy algorithms \GBR performs best with respect to quality
with an average of $93.7\%$, but the difference to the other greedy
algorithms is small. Even the greedy algorithm \GM with the lowest
total activity produces solution with an average of 91.8\% of the
optimal solution. Each of the combinations of \QAPX with subsequent
execution of a greedy algorithm outperforms each of the greedy
algorithms alone in terms of quality. However, since the solutions
produced by the greedy algorithms are already very close to the
optimal solution, we observe only a slight increase in total activity
for \QGM, \QGLC, and \QGBR over the greedy algorithms. The difference
between both approaches becomes much more visible when considering the
running time. While the average running time for the three greedy
algorithms is between 2.5ms and 3.9ms, the average running time for
the \quarter-approximation algorithms is roughly 46s. However, we note
that this large difference is mostly caused by one instance, which
required over 664s to find the solution. The median running time for
the enriched \quarter-approximation algorithms is about~1.08s.

\paragraph{Hard-conflict model.}
In Figure~\ref{fig:exp:runtime:countries:hc} we present a general
overview of the performance of the evaluated algorithms in the
1R hard-conflict model. Again, we
omit \QGLC, and \QGM in the figure to increase readability.

The performance characteristics of the algorithms in this model
resembles those in the soft-conflict model closely. The greedy
algorithms perform, again, very well with respect to both running time
and total activity.  The combination of \quarter-approximation and
greedy heuristics outperforms the greedy heuristics in total activity
again, but is also significantly slower.  Only the performance of ILP
differs much from its performance in the soft-conflict model.  To
obtain an optimal solution the ILP required on average 114s, with a
median of only 11.32s (compared to an average of 758s with a median of
30s). This is most likely caused by a decreased size in solution
space, since in the hard-conflict model, some of the ILP's activity
variables are already initially required to be set to~0.

We again discuss the two most promising approaches in more detailed,
but this time for the hard-conflict model, see Figure~\ref{fig:exp:runtime:countries:hcmodel}. Among the three greedy
heuristics the simplest algorithm (\ie, \GM) performs best with
respect to quality with an average of $96\%$, but as a whole the
greedy algorithms perform similarly well. Even the worst greedy
algorithm \GLC produces solutions of high quality (average is
95\%). Each of the combinations of the \quarter-approximation with
subsequent execution of a greedy algorithm outperforms each of the
greedy algorithms alone. However, since the solutions produced by the
greedy algorithms are so close to the optimal solution (even closer
than in the soft conflict model), we observe only a slight increase in
the quality of the solution for \QGM, \QGLC, and \QGBR. Again, the
difference is much more noticeable when considering the running time.
While the average running time for the three greedy algorithms is
between 7ms and 13ms, the average running time for the
\quarter-approximation is roughly 1.8s with a median of 0.96s. We note
an increase in the running time of the greedy algorithms compared to
the soft-conflict model. However, all three algorithms are still very
fast and suitable for real-time applications. For the combination of
the \quarter-approximation and the greedy algorithms, with the
exception of one instance, the running time of the
\quarter-approximation in the 1R hard-conflict model is very similar
to the running time in the 1R soft-conflict model. This might be
surprising since the running time of the \quarter-approximation is
dominated by obtaining optimal solutions for all of the subinstances
by the ILP, which performs better in the hard-conflict model. We
conjecture that this is because the subinstances required for the
\quarter-approximation are very small and thus the performance
increase is not that noticeable.

\begin{figure}[tb]
  \begin{subfigure}[t]{.55\textwidth}
  \centering
  \includegraphics[height = 3.5cm, trim = 12 15 20 10, clip]{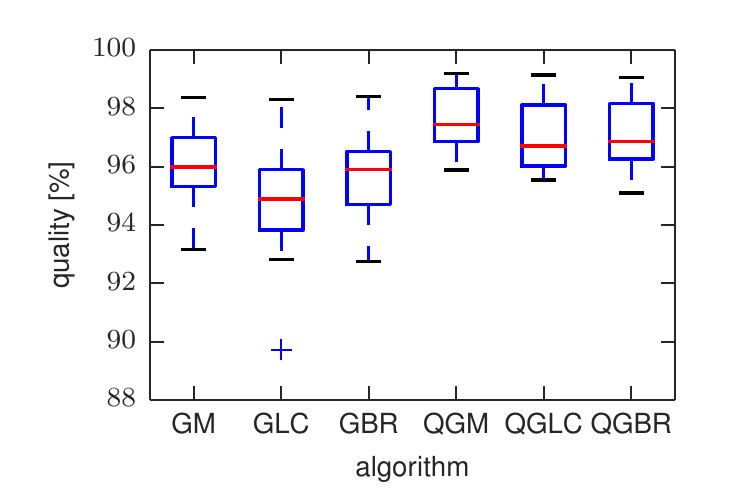} 
  \caption{Quality of the solutions as a percentage of the (optimal)
    ILP solution.}
  \label{fig:exp:runtime:countries:quality:hc}
  \end{subfigure}
  \hfil
  \begin{subfigure}[t]{.39\textwidth}
  \centering
  \includegraphics[height = 3.5cm, trim= 10 15 20 10, clip]{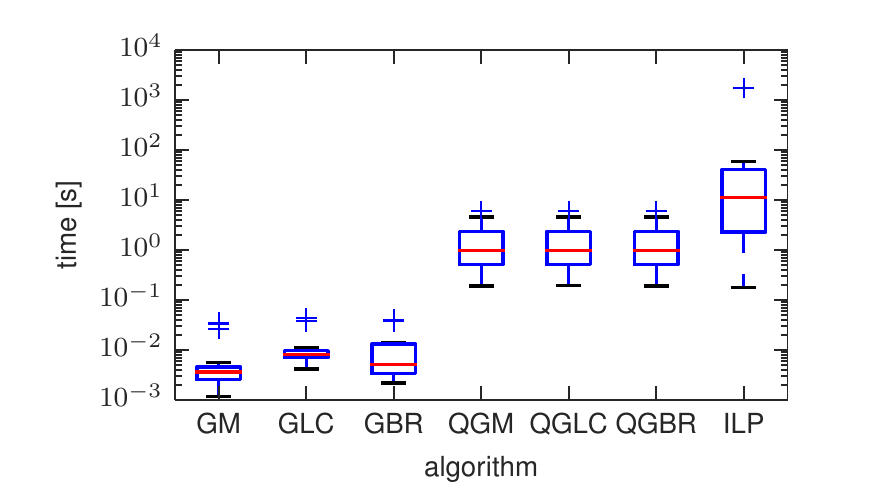} 
  \caption{Runtime (log scale) of the algorithms.}
  \label{fig:exp:runtime:countries:time:hc}
  \end{subfigure}

	\caption{Performance of the greedy algorithms and \QAPX with
          greedy postprocessing in the 1R hard-conflict model.}
  \label{fig:exp:runtime:countries:hcmodel}
\end{figure}

\paragraph{Additional Instances.}
Recall that we presented besides the set of country instances also the
city instances, which we did not use for evaluation so far. Note that
these instances (described in Table~\ref{tab:exp:input:restaurants})
are much larger than the country instances. Unfortunately, for almost
all of these instances we were not able to obtain optimal solutions in
the 1R-model with the ILP in a reasonable time frame. We therefore
cannot compare the quality of the obtained solutions by the greedy
algorithms, and by the combination of \quarter-approximation and
greedy algorithms with the optimal solution. Instead, we chose to
compare the results of the algorithms with the optimal solution in the
$\infty$R-model, which serves as an upper bound on the optimal
solution in the 1R-model; see
Table~\ref{tab:appendix:eval:restaurants}.

\begin{table}[tb]
\center
\caption{Evaluation of the three greedy algorithms, and the combination
  of \quarter-approximation with the greedy algorithms for the
  restaurant instances; see Table~\ref{tab:exp:input:restaurants}. We
  compare the total activity of the solutions with the optimal
  solution in the
  $\infty$R-model. Table~\subref{tab:appendix:eval:restaurants:sc}
  reports the results for the 1R soft-conflict model, and
  Table~\subref{tab:appendix:eval:restaurants:hc} reports the results
  for the 1R hard-conflict.}
\label{tab:appendix:eval:restaurants}
\begin{subtable}{\textwidth}
\center
\caption{1R soft-conflict model.}
\label{tab:appendix:eval:restaurants:sc}
\begin{tabular}{@{}lrrrrrrrrrrrr@{}}
\toprule
 & \multicolumn{11}{c}{\emph{algorithms}} \\
 \cmidrule(r){2-13}
 & \multicolumn{2}{c}{GM} & \multicolumn{2}{c}{GLC} & \multicolumn{2}{c}{GBR} & \multicolumn{2}{c}{QGM} & \multicolumn{2}{c}{QGLC} & \multicolumn{2}{c}{QGBR}\\
\cmidrule(r){2-3} \cmidrule(r){4-5} \cmidrule(r){6-7} \cmidrule(r){8-9} \cmidrule(r){10-11} \cmidrule(r){12-13}
\emph{scales} & time & \totact & time & \totact & time & \totact &
time & \totact & time & \totact & time & \totact \\\midrule
20m &  0.03s & 97.69\% & 0.03s & 98.01\% & 0.03s & 98.58\% & 1.30s & 98.21\% & 1.30s & 98.44\% & 1.31s & 98.70\% \\
50m &  0.03s & 93.88\% & 0.04s & 94.45\% & 0.04s & 95.52\% & 238.50s & 95.04\% & 238.52s & 95.52\% & 238.52s & 95.82\% \\
100m &  0.03s & 89.36\% & 0.05s & 89.90\% & 0.06s & 91.39\% & 1124.23s & 91.18\% & 1124.28s & 91.33\% & 1124.29s & 92.02\% \\
\bottomrule
\end{tabular}
\end{subtable}

\begin{subtable}{\textwidth}
\center
\caption{1R hard-conflict model.}
\label{tab:appendix:eval:restaurants:hc}
\begin{tabular}{@{}lrrrrrrrrrrrr@{}}
\toprule
 & \multicolumn{11}{c}{\emph{algorithms}} \\
 \cmidrule(r){2-13}
 & \multicolumn{2}{c}{GM} & \multicolumn{2}{c}{GLC} &
 \multicolumn{2}{c}{GBR} & \multicolumn{2}{c}{QGM} &
 \multicolumn{2}{c}{QGLC} & \multicolumn{2}{c}{QGBR}\\
\cmidrule(r){2-3}
\cmidrule(r){4-5}
\cmidrule(r){6-7}
\cmidrule(r){8-9}
\cmidrule(r){10-11}
\cmidrule(r){12-13}
\emph{scales} & time & \totact & time & \totact & time & \totact &
time & \totact & time & \totact & time & \totact \\\midrule
20m &  0.72s & 97.23\% & 0.80s & 97.20\% & 0.80s & 97.39\% & 4.29s & 97.53\% & 4.30s & 97.55\% & 4.29s & 97.58\% \\
50m &  0.79s & 92.45\% & 0.86s & 92.21\% & 0.87s & 92.53\% & 16.91s & 92.99\% & 16.94s & 92.93\% & 16.94s & 92.92\% \\
100m &  0.78s & 87.05\% & 0.85s & 86.09\% & 0.86s & 86.61\% & 54.50s &
87.56\% & 54.56s & 86.96\% & 54.58s & 87.17\% \\
\bottomrule
\end{tabular}
\end{subtable}
\end{table}

We observe that the results for all algorithms in both models are very
close to the maximal achievable total activity in the $\infty$R
model. Even in the most dense instance in the 1R hard-conflict model,
the worst greedy algorithm achieves still on average about 86\% of the
maximal achievable total activity with a running time of about
0.8s. The quality of the greedy algorithms is slightly better in the
1R soft-conflict model and the running time is much smaller (around
30--40ms).  The combination of \quarter-approximation and greedy
heuristic does perform slightly better, but the difference is
minuscule. However, the running time differs drastically. While in the
country instances before, the average running time was less than 1s in
both the 1R hard-conflict and the 1R soft-conflict model, the running
time for the combination of \quarter-approximation and greedy
heuristic is much slower. For the most dense instances the running
time is on average one minute in the hard-conflict model and for the
soft-conflict model about 20 minutes. However, at this point we need
to mention that the running times for the instances vary quite
significantly and the only conclusion we can draw safely from the
detailed data is that this approach is much slower (even for the
easiest instance the approach required for the soft-conflict model
about 1 minute to compute the solution), but drawing conclusions
beyond this is impossible.

The results indicate that the greedy algorithms in the 1R
soft-conflict model produce solutions that are very close to the
maximum possible total activity while the algorithms require only few
milliseconds in running time.  This strengthens our conclusion that the
1R  model in combination with any of the three greedy
algorithms is the best strategy for labeling rotating maps. Whether to use the soft-conflict model or the hard-conflict model is a design choice that should depend on the requirements of the actual application.

In order to give a final recommendation for an algorithm, it is
necessary to make a choice on the time-quality trade-off that is
acceptable in a particular situation. If running time is not the
primary concern, e.g., for offline applications with high computing
power available, we can recommend the ILP, which ran reasonably fast
in our experiments, at least for the smaller country instances. On the other hand, if computing power is limited, instances are large,
or real-time labeling is necessary, e.g., on a mobile device, all
three greedy heuristics can be recommended as the methods of choice; a
slight advantage of \GBR was observed in our experiments. All three
algorithms run very fast (a few milliseconds) and empirically produce
high activities of more than $90\%$ of the optimum solution. If one
wants to invest some extra running time, the combination of
\QAPX with a greedy algorithm may be of interest as it produces
slightly better solutions than the stand-alone greedy algorithms and is much faster than the ILP.

\section{Conclusion}
\label{sec:conclusion}
In this work, we evaluated different strategies for labeling dynamic
maps that allow continuous rotation, where a labeling strategy consists of a
consistency model and a labeling algorithm. In the first part of the evaluation, we considered the quality of optimal solutions in different consistency models. It turned out that the restriction to  one or two active ranges per label (1R- and 2R-models) yields the best compromise in terms of low flickering and high total activity value of more than 95\% of the upper bound obtained from the unrestricted model ($\infty$R). Additionally, treating all pairwise label conflicts as soft conflicts increased the total activity values between 8\% and 32\% at the cost of occasional occlusion of unlabeled point features.

In the second part of the evaluation, we investigated the performance
of three new greedy heuristics and our previous \quarter-approximation
algorithm~\cite{gnr-clrm-w-11} in terms of labeling quality and
running time. It turned out that the greedy heuristics performed very
well in both total activity (well above 90\%) and running time (a few
ms). The unmodified \quarter-approximation performs much worse, but
the combination of \quarter-approximation and greedy heuristics yields
slightly higher total activity than the greedy heuristics alone; the
running time, however, can grow to several seconds. In conclusion, we
believe that the 1R model in combination with any of the three greedy
algorithms is, in most cases, the best labeling strategy for labeling
dynamic rotating maps. Whether the soft-conflict  or the
hard-conflict model is more appropriate depends on requirements of the application.

\subsection*{Acknowledgments}

We thank an anonymous reviewer of the SEA 2014 version of this paper for the helpful and detailed comments. This work was partially supported by a Google Research Award.


\begin{thebibliography}{10}

\bibitem{aks-lpmis-98}
P.~K. Agarwal, M.~van Kreveld, and S.~Suri.
\newblock {Label Placement by Maximum Independent Set in Rectangles}.
\newblock {\em Comput. Geom. Theory Appl.}, 11:209--218, 1998.

\bibitem{bdy-dml-06}
K.~Been, E.~Daiches, and C.~Yap.
\newblock Dynamic map labeling.
\newblock {\em {IEEE} Trans. Visualization and Computer Graphics},
  12(5):773--780, 2006.

\bibitem{bnpw-oarcd-10}
K.~Been, M.~Nöllenburg, S.-H. Poon, and A.~Wolff.
\newblock Optimizing active ranges for consistent dynamic map labeling.
\newblock {\em Comput. Geom. Theory Appl.}, 43(3):312--328, 2010.

\bibitem{bg-lmpwtblslo-13}
M.~Berg and D.~H. Gerrits.
\newblock Labeling moving points with a trade-off between label speed and label
  overlap.
\newblock In H.~Bodlaender and G.~Italiano, editors, {\em Proc. 21th Annu.
  European Symp. Algorithms (ESA'13)}, volume 8125 of {\em Lecture Notes
  Comput. Sci.}, pages 373--384. Springer Verlag, 2013.

\bibitem{bg-aafm-12}
M.~de~Berg and D.~H.~P. Gerrits.
\newblock Approximation algorithms for free-label maximization.
\newblock {\em Comput. Geom. Theory Appl.}, 45(4):153--168, 2012.

\bibitem{fw-ppalm-91}
M.~Formann and F.~Wagner.
\newblock A packing problem with applications to lettering of maps.
\newblock In {\em Proc. 7th Ann. ACM Symp. Comput. Geom. (SoCG'91)}, pages
  281--288, 1991.

\bibitem{gnn-tbdml-isaac-13}
A.~Gemsa, B.~Niedermann, and M.~N{\"o}llenburg.
\newblock {Trajectory-Based Dynamic Map Labeling}.
\newblock In {\em Proc. 24th Ann. Internat. Symp. Alg. and Comput. (ISAAC'13)},
  volume 8283 of {\em Lecture Notes Comput. Sci.} Springer Verlag, 2013.
\newblock Full version available at \url{http://arxiv.org/abs/1309.3963}.

\bibitem{gnr-clrm-w-11}
A.~Gemsa, M.~N{\"o}llenburg, and I.~Rutter.
\newblock {Consistent Labeling of Rotating Maps}.
\newblock In F.~Dehne, J.~Iacono, and J.-R. Sack, editors, {\em Proc. 12th
  Internat. Workshop on Algorithms and Data Structures (WADS'11)}, volume 6844
  of {\em Lecture Notes Comput. Sci.}, pages 451--462. Springer Verlag, August
  2011.
\newblock Full version available at \url{http://arxiv.org/abs/1104.5634}.

\bibitem{hm-ascpp-85}
D.~S. Hochbaum and W.~Maass.
\newblock Approximation schemes for covering and packing problems in image
  processing and {VLSI}.
\newblock {\em J. ACM}, 32:130--136, 1985.

\bibitem{nps-dosbl-10}
M.~N{\"o}llenburg, V.~Polishchuk, and M.~Sysikaski.
\newblock Dynamic one-sided boundary labeling.
\newblock In {\em Proc. 18th ACM SIGSPATIAL Internat. Conf. Adv. in Geo. Inf.
  Sys. (ACM GIS'10)}, pages 310--319, 2010.

\bibitem{okf-dlu-09}
K.~Ooms, W.~Kellens, and V.~Fack.
\newblock Dynamic map labeling for users.
\newblock In {\em Proc. 24th Internat. Cartographic Conf. (ICC'09)}, 2009.

\bibitem{vtw-tcrld-2012}
M.~Vaaraniemi, M.~Treib, and R.~Westermann.
\newblock Temporally coherent real-time labeling of dynamic scenes.
\newblock In {\em Proc. 3rd Internat. Conf. on Computing for Geospatial
  Research and Applications}, COM.Geo '12, pages 17:1--17:10. ACM, 2012.

\bibitem{ksw-pslsl-99}
M.~van Kreveld, T.~Strijk, and A.~Wolff.
\newblock Point labeling with sliding labels.
\newblock {\em Comput. Geom. Theory Appl.}, 13:21--47, 1999.

\bibitem{yi-ptalsmrm-13}
Y.~Yokosuka and K.~Imai.
\newblock Polynomial time algorithms for label size maximization on rotating
  maps.
\newblock pages 187--192, 2013.

\end{thebibliography}
\end{document}